\RequirePackage{fix-cm}
\documentclass[11pt]{article}
\usepackage{cite}
\usepackage{epstopdf}
\usepackage{graphicx}
\usepackage{subfigure}
\usepackage{geometry}
\usepackage[affil-it]{authblk}
\usepackage{amsmath}
\usepackage{amsthm}
\usepackage{amssymb}
\usepackage{url}
\usepackage{thmtools}
\usepackage[linesnumbered,ruled,vlined]{algorithm2e}
\usepackage{mathtools}
\usepackage{color}
\usepackage{fixmath}
\usepackage{array}
\usepackage{multirow}
\usepackage{etoolbox}
\newtheorem{lemma}{Lemma}
\newtheorem{theorem}{Theorem}
\newtheorem{corollary}{Corollary}
\newtheorem{myclaim}{Claim}{\bf}{\it}
\newtheorem{definition}{Definition}{\bf}{\it}
\newtheorem{open}{Open Question}{\bf}{\it}

\DeclareMathOperator*{\argmax}{arg\,max}
\providecommand{\keywords}[1]{\textbf{\textit{Keywords---}} #1}

\begin{document}
\title{On the Approximability and Hardness of the Minimum Connected Dominating 
Set with Routing Cost Constraint}
\author{Tung-Wei Kuo%
  \thanks{E-mail: \texttt{twkuo@cs.nccu.edu.tw}}
  }
\affil{Department of Computer Science, National Chengchi University}

\date{Dated: \today}
\maketitle

\begin{abstract}
In the problem of minimum connected dominating set with routing cost constraint, 
we are given a graph $G=(V,E)$, and the goal is to find the smallest connected 
dominating set $D$ of $G$ such that, 
for any two non-adjacent vertices $u$ and $v$ in $G$, 
the number of internal nodes on the shortest path between 
$u$ and $v$ in the subgraph of $G$ induced by $D \cup \{u,v\}$ is at most 
$\alpha$ times that in $G$.  
For general graphs, the only known previous approximability result 
is an $O(\log n)$-approximation algorithm ($n=|V|$) 
for $\alpha = 1$ by Ding \textit{et al.} 
For any constant $\alpha > 1$, 
we give an $O(n^{1-\frac{1}{\alpha}}(\log n)^{\frac{1}{\alpha}})$-approximation 
algorithm.
When $\alpha \geq 5$, we give an $O(\sqrt{n}\log n)$-approximation algorithm.
Finally, we prove that, when $\alpha =2$, 
unless $NP \subseteq DTIME(n^{poly\log n})$, for any constant $\epsilon > 0$,
the problem admits no polynomial-time
$2^{\log^{1-\epsilon}n}$-approximation algorithm, 
improving upon the $\Omega(\log n)$ bound by 
Du \textit{et al.} (albeit under
a stronger hardness assumption).  
\end{abstract}
\keywords{Connected dominating set, spanner, set cover with pairs, MIN-REP problem}

\section{Introduction}
\subsection{Motivation}
In wireless network routing, a common approach is to 
select a set of nodes as the \textit{virtual backbone}. 
The virtual backbone is responsible for relaying packets.
Specifically, when a node $s$ generates a packet destined to $d$, 
the packet is routed through path $(s, v_1, v_2, \cdots, v_k, t)$, 
where every internal node $v_i, 1 \leq i \leq k,$ belongs to the virtual backbone.
To realize this idea, we can model the wireless network as a graph 
$G=(V,E)$, where $V$ is the set of nodes in the wireless network, 
and $(u,v) \in E$ if and only if $u$ and $v$ can communicate with each other directly. 
Thus, a connected dominating set of $G$ is a virtual backbone for 
the wireless network.\footnote{A set $D \subseteq V$ is a \textit{dominating set} 
of $G=(V,E)$ if every vertex in $V \setminus D$ is adjacent to $D$. 
Furthermore, if $D$ induces a connected subgraph of $G$, 
then $D$ is called a \textit{connected dominating set} of $G$.} 
One of the concerns in constructing the
virtual backbone is the routing cost. Specifically, the routing cost of 
sending a packet from the source $s$ to the destination $d$ is the 
number of internal nodes (relays) in the routing path from $s$ to $d$.
For example, the routing cost is $k$ if the routing path is 
$(s, v_1, v_2, \cdots, v_k, t)$. 
The routing cost should not be too high even if
packets are only allowed to be routed through the virtual backbone. 
Next, we give the formal definition of the problem. 

\subsection{Problem Definition}
Let $G[S]$ be the subgraph of $G=(V,E)$ induced by $S \subseteq V$.
Let $m_G(u,v)$ be the number of internal vertices on the shortest path 
between $u$ and $v$ in $G$. For example, if $u$ and $v$ are adjacent, 
then $m_G(u,v) = 0$.
If $u$ and $v$ are not adjacent and have a common neighbor, 
then $m_G(u,v) = 1$. 
Furthermore, given a vertex subset $D$ of $G$, 
$m^D_G(u,v)$ is defined as $m_{G[D \cup \{u,v\}]}(u,v)$, 
i.e., the number of internal vertices on the shortest path between $u$ 
and $v$ through $D$.
We use $n(G)$ to denote the number of vertices in graph $G$.
When the graph we are referring to is clear from the context, 
we simply write $n$, $m(u,v)$, and $m^D(u,v)$ instead of $n(G)$, $m_G(u,v)$, 
and $m^D_G(u,v)$, respectively.

\begin{definition}
Given a connected graph $G$ and a positive integer $\alpha$, 
the \textbf{Connected Dominating set problem with 
Routing cost constraint} (CDR-$\alpha$) asks for the smallest 
connected dominating set $D$ of $G$, 
such that, for every two vertices $u$ and $v$, 
if $u$ and $v$ are not adjacent in $G$, 
then $m^D(u,v) \leq \alpha \cdot m(u,v)$.
\end{definition}

\subsection{Preliminary}
\subsubsection{An Equivalent Problem}
In the CDR-$\alpha$ problem, we need to consider all the pairs of non-adjacent 
nodes. Ding \textit{et al.} discovered that to solve the CDR-$\alpha$ problem,
it suffices to consider only vertex pairs $(u,v)$ such that $m(u,v)= 1$, 
i.e., $u$ and $v$ are not adjacent but have a common neighbor~\cite{5703070}. 
We call the corresponding problem the 1-DR-$\alpha$ problem.

\begin{definition}
Given a connected graph $G = (V,E)$ and a positive integer $\alpha$, 
the 1-DR-$\alpha$ problem asks for the smallest dominating set 
$D$ of $G$, such that, for every two vertices $u$ and $v$, if $m(u,v) = 1$, 
then $m^D(u,v) \leq \alpha$.
\end{definition}
We say that $u$ and $v$ form a \textbf{target couple}, denoted by $[u,v]$, 
if $m(u,v) = 1$. We say that a set $S$ \textbf{covers} a target couple 
$[u,v]$ if $m^S(u,v) \leq \alpha$. Hence, the 1-DR-$\alpha$ problem asks for the 
smallest dominating set that covers all the target couples. Notice that any feasible 
solution of the 1-DR-$\alpha$ problem must induce a connected subgraph of $G$.
The equivalence between the CDR-$\alpha$ problem and the 1-DR-$\alpha$ 
problem is stated in the following theorem.

\begin{theorem}[Ding \textit{et al.} \cite{5703070}]
$D$ is a feasible solution of the CDR-$\alpha$ problem with input graph $G$ 
if and only if $D$ is a feasible solution of the 1-DR-$\alpha$ problem 
with input graph $G$.
\end{theorem}

\begin{corollary}
Any $r$-approximation algorithm of the 1-DR-$\alpha$ problem 
is an $r$-approximation algorithm of the CDR-$\alpha$ problem.
\end{corollary}
In this paper, we thus focus on the 1-DR-$\alpha$ problem.

\subsubsection{Feasibility of the 1-DR-$\alpha$ Problem for $\alpha \geq 5$}
Next, we give the basic idea of finding a feasible solution of
the 1-DR-$\alpha$ problem for $\alpha \geq 5$ 
used in previous researches, e.g., in~\cite{Du2010}.
One of our algorithms still uses this idea.
First, find a dominating set $D$. 
Thus, for any target couple $[u,v]$, there exist $u^d$ and $v^d$ in $D$, 
such that $u^d$ and $v^d$ dominate $u$ and $v$, respectively.\footnote{$u^d$ dominates 
$u$ if $u^d = u$ or $u^d$ and $u$ are adjacent.}
Let $D' = D$. For any two vertices $u'$ and $v'$ in $D$, 
if $m(u',v') \leq 3$, then we add the $m(u',v')$ internal vertices 
of the shortest path between $u'$ and $v'$ on $G$ to $D'$.
Observe that $m(u^d,v^d) \leq 3$. 
Hence, $m^{D'}(u,v) \leq 5$ and $D'$ is a feasible solution of the 
1-DR-$\alpha$ problem for $\alpha \geq 5$.

\begin{lemma}
\label{idea}
Let $D$ be a dominating set of $G$. 
Let $D' \supseteq D$ be a vertex subset of $G$ such that, 
for any two vertices $u'$ and $v'$ in $D$, 
if $m(u',v') \leq 3$, then $m^{D'}(u',v') \leq 3$.
Then, $D'$ is a feasible solution of the 1-DR-$\alpha$ 
problem with input $G$ and $\alpha \geq 5$. 
\end{lemma}

\subsection{Previous Result}
\paragraph*{Previous result on general graphs}
When $\alpha = 1$, the 1-DR-$\alpha$ problem can be transformed 
to the set cover problem, i.e., cover all the vertices (to form a dominating set) 
and cover all the target couples. Observe that each target couple can be covered by 
a single vertex. The resulting approximation ratio is 
$O(\log n)$~\cite{5703070}. When $\alpha$ is sufficiently large, 
e.g., $\alpha \geq n$, any connected dominating set is feasible for the 
CDR-$\alpha$ problem. Note that, for any $\alpha$, the size of the minimum 
connected dominating set is a lower bound of the CDR-$\alpha$ problem. 
Since the connected dominating set can be approximated 
within a factor of $O(\log n)$~\cite{Guha1998, RUAN2004325}, 
the CDR-$n$ problem can be 
approximated within a factor of $O(\log n)$. 
If $\alpha$ falls between these two extremes, 
e.g., $\alpha = 2$, the only known previous result is the trivial 
$O(n)$-approximation algorithm. On the hardness side, 
it has been proved that, unless $NP \subseteq DTIME(n^{\log \log n})$, 
there is no polynomial-time algorithm that can approximate 
the CDR-$\alpha$ problem within 
a factor of $\rho \ln \delta$ ($\forall \rho<1$) for 
$\alpha = 1$~\cite{5703070} and $\alpha \geq 2$~\cite{Du2013, 6216366}, 
where $\delta$ is the maximum degree of $G$.

\begin{open}[Du and Wan~\cite{Du2013}]
Is there a polynomial-time $O(\log n)$-approximation algorithm for the 
CDR-$\alpha$ problem for $\alpha \geq 2$?
\end{open}

\paragraph*{Previous result on Unit Disk Graph (UDG)}
Most of the studies on the CDR-$\alpha$ problem focused on UDG~\cite{5703070, 
Du2013, 6216366, Du2010, 7524455}. 
UDG exhibits many nice properties that 
enable constant factor approximation algorithms (or PTAS) 
in many problems where only 
$O(\log n)$-approximation algorithms (or worse) are known in general graphs, 
e.g., the minimum (connected) dominating set problem and the maximum 
independent set problem~\cite{NET:NET10097, DAS2015439, Nieberg2006}. 
All the previous research on the CDR-$\alpha$ problem on UDG
leveraged constant bounds of the maximum independent set or 
the minimum dominating set.
However, all the previous research only solved the case where 
$\alpha \geq 5$ (by Lemma~\ref{idea}),
and the best result so far is a PTAS by Du~\textit{et al.}~\cite{Du2010}.
When $1 < \alpha < 5$, the only known previous result 
is the trivial $O(n)$-approximation algorithm.

\subsection{Our Result and Basic Ideas}
In this paper, we first give an approximation algorithm of the 
1-DR-$\alpha$ problem on general graphs for any constant $\alpha > 1$. 
A critical observation is that the 1-DR-$2$ problem is a special 
case of the Set Cover with Pairs (SCP) problem~\cite{Hassin2005}. 
Hassin and Segev
proposed an $O(\sqrt{t\log t})$-approximation algorithm for the SCP problem, 
where $t$ is the number of targets to be covered. 
However, since there are $O(n^2)$ target couples to be covered, 
directly applying the $O(\sqrt{t\log t})$-approximation bound 
yields a trivial upper bound for the 1-DR-2 problem. 
We re-examine the analysis in~\cite{Hassin2005} and 
find that, when applying the algorithm to the 1-DR-2 problem, 
the approximation ratio can also be expressed as $O(\sqrt{n\log n})$. 
Nevertheless, in this paper, we give a slightly simplified algorithm 
with an easier analysis for the SCP problem. 
The algorithm and analysis also make it easy to solve the generalized SCP problem.
We obtain the following result, which is the first non-trivial result of the 
CDR-$\alpha$ problem for $\alpha > 1$ 
on general graphs and for $1 < \alpha < 5$ on UDG. 
\begin{theorem}\label{thrm: 1st}
For any constant $\alpha > 1$, there is an
$O(n^{1-\frac{1}{\alpha}}(\log n)^{\frac{1}{\alpha}})$-approximation 
algorithm for the 1-DR-$\alpha$ problem. 
\end{theorem}

Apparently, the above performance guarantee deteriorates quickly 
as $\alpha$ increases.
In our second algorithm, we apply the aforementioned idea of finding a 
feasible solution when $\alpha \geq 5$, i.e., Lemma~\ref{idea}. 
We have the following result.
\begin{theorem}\label{thrm: 2nd}
When $\alpha \geq 5$, 
there is an $O(\sqrt{n}\log n)$-approximation 
algorithm for the 1-DR-$\alpha$ problem. 
\end{theorem}

Finally, we answer Open Question 1 negatively.
We improve upon the $\Omega(\log n)$ hardness result  
for the 1-DR-2 problem 
(albeit under a stronger hardness assumption)~\cite{Du2013, 6216366}.
In this paper, we give a reduction from the MIN-REP problem~\cite{Kortsarz2001}. 

\begin{theorem}
\label{thrm: inapprox}
Unless $NP \subseteq DTIME(n^{poly\log n})$, for any constant 
$\epsilon > 0$, the 1-DR-2 problem admits 
no polynomial-time
$2^{\log^{1-\epsilon}n}$-approximation algorithm,
even if the graph is triangle-free\footnote{If the graph is triangle-free, 
then any two vertices with a common neighbor form a target couple.} 
or the constraint that the feasible 
solution must be a dominating set is 
ignored\footnote{One may drop the 
constraint that the solution must be a 
dominating set, and focuses on minimizing the number of vertices to 
cover all the target couples. This theorem also applies to 
such a problem.}.
\end{theorem}

\subsection{Relation with the Basic $k$-Spanner Problem}
When we ignore the constraint that any feasible solution must be a 
connected dominating set, 
the CDR-$\alpha$ problem is similar to the basic $k$-spanner problem.
For completeness, we give the formal definition of the basic $k$-spanner problem.
Given a graph $G = (V,E)$, a $k$-spanner of $G$ is a subgraph $H$ of $G$ 
such that $d_H(u,v) \leq kd_G(u,v)$ for all $u$ and $v$ in $V$, 
where $d_G(u,v)$ is the number of edges in the shortest path between 
$u$ and $v$ in $G$. 
The basic $k$-spanner problem asks for the $k$-spanner that has the fewest 
edges. The CDR-$\alpha$ problem differs with the 
basic $k$-spanner problem in the following three aspects: First,
in the CDR-$\alpha$ problem, we find a set of vertices $D$, 
and all the edges in the subgraph induced by $D$ can be used for routing; 
while in the basic $k$-spanner problem, only edges in $H$ can be used.
Second,  in the CDR-$\alpha$ problem, the objective is to minimize the 
number of chosen vertices; while in the basic $k$-spanner problem, the 
objective is to minimize the number of chosen edges.
Finally, in the basic $k$-spanner problem, the distance is measured by 
the number of edges; while in the CDR-$\alpha$ problem, the distance is 
measured by the number of internal nodes. Despite the above differences, these 
two problems share similar approximability and hardness results. 
Alth\"{o}fer \textit{et al.} proved that 
every graph has a $k$-spanner of at most 
$n^{1+\frac{1}{\lfloor (k+1)/2\rfloor}}$ edges, 
and such a $k$-spanner can be constructed in polynomial time~\cite{Althofer1993, 
Dinitz:2016:ALS:2884435.2884494}.
Since the number of edges in any $k$-spanner is at least $n-1$, 
this yields an 
$O(n^{\frac{1}{\lfloor (k+1)/2\rfloor}})$-approximation 
algorithm for the basic $k$-spanner problem. 
For $k = 2$, there is an $O(\log n)$-approximation algorithm due to Kortsarz and 
Peleg~\cite{KORTSARZ1994222}, and this is the best 
possible~\cite{Kortsarz2001}.
For $k = 3$, Berman \textit{et al.} proposed an 
$\tilde{O}(n^{1/3})$-approximation 
algorithm~\cite{BERMAN201393}. For $k = 4$, Dinitz and Zhang proposed an
$\tilde{O}(n^{1/3})$-approximation 
algorithm~\cite{Dinitz:2016:ALS:2884435.2884494}.
On the hardness side, it has been proved that for any constant $\epsilon > 0$
and for $3 \leq k \leq \log^{1-2\epsilon}n$, 
unless $NP \subseteq BPTIME(2^{poly\log n})$,
there is no polynomial-time algorithm that approximates the basic
$k$-spanner problem to a factor better than 
$2^{(\log^{1-\epsilon}n)/k}$~\cite{Dinitz:2015:LCI:2846106.2818375}.

\section{Two Algorithms for the 1-DR-$\alpha$ Problem}
\subsection{The First Algorithm}
We first give the formal definition of the 
Set Cover with Pairs (SCP) problem. 

\begin{definition}
Let $T$ be a set of $t$ targets. 
Let $V$ be a set of $n$ elements.
For every pair of elements $P=\{v_1, v_2\} \subseteq V$,
$C(P)$ denotes the set of targets covered by $P$.
The Set Cover with Pairs (SCP) problem asks for the smallest 
subset $S$ of $V$ such that 
$\bigcup\limits_{\{v_1, v_2\} \subseteq S}{C(\{v_1, v_2\})}=T$.
\end{definition}
Let $OPT$ be the number of elements in the optimal solution.
We only need to consider the case where $t > 1$ and $OPT > 1$.   
\subsubsection{Approximating the SCP Problem}
\label{alg: SCP}
Our algorithm is a simple greedy algorithm: 
in each round, we choose at most two elements $u$ and $v$ that 
maximize the number of covered targets. 
Specifically, $S$ is an empty set initially.
In each round, we select a set $P \subseteq V\setminus S$
such that $|P| \leq 2$ and $P$ increases 
the number of covered targets the most, i.e.,
$P=\argmax\limits_{P': |P'| \leq 2, P' \subseteq V \setminus S} {g(P')}$, where 
$$g(P') = |\bigcup_{\{v_1,v_2\} \subseteq S \cup P}{C(\{v_1,v_2\})}|
-|\bigcup_{\{v_1,v_2\} \subseteq S}{C(\{v_1,v_2\})}|.$$
We then add $P$ to $S$ and repeat the above process until
all the targets are covered.
The algorithm terminates once all targets are covered.\footnote{In~\cite{Hassin2005}, 
in each round, a set $P=\argmax\limits_{P': |P'| \leq 2, P' \subseteq V \setminus S}
{g'(P')}$ is added to $S$, where $g'(P') = \frac{g(P')}{|P'|}$.}

\begin{theorem}
\label{SCP}
The above algorithm is an $O(\sqrt{n\log t})$-approximation algorithm 
for the SCP problem.
\end{theorem}

\begin{proof}
Let $R_i$ be the number of uncovered targets after round $i$.
In the first round, some pair of elements in the optimal solution 
can cover at least $t/{{OPT}\choose{2}}$ targets. 
Since we choose a pair of elements greedily in each round, 
$R_1 \leq t(1-1/{{OPT}\choose{2}})$. 
In the second round, there exists 
a pair of elements in the optimal solution that can cover at least 
$R_1/{{OPT}\choose{2}}$ targets among the $R_1$ uncovered targets.
Again, we choose the pair of elements that 
increases the number of covered targets the most.
Hence, $R_2 \leq R_1 - R_1/{{OPT}\choose{2}}\leq 
t(1-1/{{OPT}\choose{2}})^2$. In general, 
$R_i \leq t(1-1/{{OPT}\choose{2}})^i$.
After $r = {{OPT}\choose{2}}\ln t$ rounds, 
the number of uncovered targets is at most
$t(1-1/{{OPT}\choose{2}})^r 
\leq t(e^{-1/{{OPT}\choose{2}}})^r
\leq te^{-\ln t} = 1$.
Hence, after $O(OPT^2\ln t)$ rounds, all targets are covered.
Let $ALG$ be the number of elements chosen by the algorithm.
Since we choose at most two elements in each round, 
$ALG = O(OPT^2\ln t)$.
Finally, since $ALG \leq n$,
$ALG = O(\sqrt{n \cdot OPT^2\ln t}) = O(\sqrt{n\ln t})OPT$.
\end{proof}

Note that, in Theorem~\ref{SCP}, 
we can replace $n$ with any upper bound of the size of solutions 
obtained by any polynomial-time algorithm $\mathcal{A}$ for the SCP problem. 
This is achieved by executing both $\mathcal{A}$ and our algorithm. 
Choosing the best between the two outputs yields the desired approximation 
ratio. An example is replacing $n$ with $2t$. 

\subsubsection{Approximating the 1-DR-2 Problem}
\label{alg: 1-DR-2}
To transform the 1-DR-2 problem to the SCP problem,
we treat each target couple as a target.
Moreover, we treat each vertex as a target so that the output is a dominating set. 
The set of elements $V$ in the SCP problem is the 
vertex set of $G$. $C(P)$ consists of all the vertices that are dominated by $P$ in $G$ 
and all the target couples that are covered by $P$ in $G$. In this SCP instance, 
$n=n(G)$ and $t=O(n(G)^2)$. 
It is easy to verify the following result.

\begin{theorem}
There is an $O(\sqrt{n\log n})$-approximation algorithm 
for the 1-DR-2 problem.
\end{theorem}

\subsubsection{The Set Cover with $\alpha$-Tuples (SCT-$\alpha$) Problem}
\label{defi: SCT}
In the 1-DR-2 problem, every target couple can be covered by no more than two vertices.
In the 1-DR-$\alpha$ problem, every target couple can be covered by no more than
$\alpha$ vertices. Hence, we consider the following generalization of 
the SCP problem.
\begin{definition}
Let $T$ be a set of $t$ targets. 
Let $V$ be a set of $n$ elements.
Let $\alpha$ be a positive integer constant greater than one.
For every $\alpha$-tuple $P=\{v_1, v_2, \cdots, v_{\alpha}\}\subseteq V$,
$C(P)$ denotes the set of targets covered by $P$.
The Set Cover with $\alpha$-Tuples (SCT-$\alpha$) problem asks for the smallest 
subset $S$ of $V$ such that 
$\bigcup\limits_{\{v_1, v_2, \cdots, v_{\alpha}\} \subseteq S}
{C(\{v_1, v_2, \cdots, v_{\alpha}\})}=T$.
\end{definition}
We only need to consider the case where $t > 1$ and $OPT \geq \alpha$ 
($\alpha$ is a constant).   
\subsubsection{Approximating the SCT-$\alpha$ Problem and the 1-DR-$\alpha$ Problem}
\label{algo: SCT}
The algorithm for the SCT-$\alpha$ problem is a straightforward generalization 
of the algorithm for the SCP problem.
The difference is that, in each round, 
we choose a set $P$ of at most $\alpha$ elements that
increases the number of covered targets the most.
The transformation from the 1-DR-$\alpha$ problem to the SCT-$\alpha$ problem is also 
similar to the previous transformation. The value of $\alpha$ in the 
constructed SCT-$\alpha$ instance is equal to that in the 1-DR-$\alpha$ instance. 
Again, $n = n(G)$ and $t = O(n(G)^2)$ in the constructed SCT-$\alpha$ instance.
Theorem~\ref{thrm: 1st} is a direct result of the following theorem.
\begin{theorem}
\label{thrm: SCT}
There is an $O(n^{1-\frac{1}{\alpha}} \cdot 
(\ln t)^{\frac{1}{\alpha}})$-approximation algorithm for the 
SCT-$\alpha$ problem.
\end{theorem}

We have the following claim, whose proof can be found in the appendix.
\begin{myclaim}
\label{c}
When $c = \frac{1}{\alpha}-\frac{\ln \ln (t^{\alpha})}{\alpha \ln n}$, 
$n^{1-c}=\sqrt{n \cdot \alpha(n^c)^{\alpha-2}\ln t} = 
n^{1-\frac{1}{\alpha}} \cdot (\alpha \ln t)^{\frac{1}{\alpha}}$.
\end{myclaim}

\textbf{Proof of Theorem~\ref{thrm: SCT}:}
Let $R_i$ be the number of uncovered targets after round $i$.
By a similar argument in the proof of Theorem~\ref{SCP}, 
we get that $R_i \leq t(1-1/{{OPT}\choose{\alpha}})^i$.
After $r = {{OPT}\choose{\alpha}}\ln t$ rounds, 
the number of uncovered targets is at most one.
Hence, after $O(OPT^{\alpha}\ln t)$ rounds, all targets are covered.
Let $ALG$ be the number of elements chosen by the algorithm.
Since we choose at most $\alpha$ elements in each round, 
$ALG = O(\alpha OPT^{\alpha}\ln t)$.
Since $ALG \leq n$, $ALG = O(\sqrt{n \cdot \alpha OPT^{\alpha}\ln t})$.

Let $c = \frac{1}{\alpha}-\frac{\ln \ln (t^{\alpha})}{\alpha \ln n}$.
When $OPT \geq n^c$, the approximation ratio is $n^{1-c}$.
When $OPT \leq n^c$, 
$ALG = O(\sqrt{n \cdot \alpha OPT^{\alpha-2}\ln t})OPT
= O(\sqrt{n \cdot \alpha(n^c)^{\alpha-2}\ln t})OPT$.
The proof then follows from Claim~\ref{c} and 
$\alpha^{\frac{1}{\alpha}}=O(1)$.
\qed

\subsection{The Second Algorithm}
The second algorithm is designed for the 1-DR-$\alpha$ problem 
when $\alpha \geq 5$. It has a better approximation ratio than 
that of the previous algorithm when $\alpha \geq 5$. 
The algorithm is suggested in Lemma~\ref{idea}: 
We first find a dominating set $D$ by any $O(\log n)$-approximation algorithm.
Let $D' = D$.
For any two vertices $u$ and $v$ in $D$, if $m(u,v) \leq 3$, we then 
add at most three vertices to $D'$ so that $m^{D'}(u,v) \leq 3$.

\textbf{Proof of Theorem~\ref{thrm: 2nd}:}
Let $OPT_{DS}$ be the size of the minimum dominating set in $G$.
Let $OPT$ be the size of the optimum of the 1-DR-$\alpha$ problem.
Since any feasible solution of the 1-DR-$\alpha$ problem must be a 
dominating set, $OPT_{DS} \leq OPT$.
$|D'| \leq |D|+3{{|D|}\choose{2}} = O((\log n \cdot OPT_{DS})^2) 
= O((\log n \cdot OPT)^2)$.
Since $|D'| \leq n$, we have $|D'| = O(\sqrt{n \cdot (\log n \cdot OPT)^2}) = 
O(\sqrt{n}\log n)OPT$. \qed

\section{Inapproximability Result}
\label{inapprox}
\subsection{The MIN-REP Problem}
We prove Theorem~\ref{thrm: inapprox} by a reduction from the 
MIN-REP problem~\cite{Kortsarz2001}. 
The input of the MIN-REP problem consists of a 
bipartite graph $G=(X, Y, E)$, a partition of $X$, 
$\mathcal{P}_X=\{X_1, X_2, \cdots, X_{k_X}\}$, and a partition of $Y$,
$\mathcal{P}_Y=\{Y_1, Y_2, \cdots, Y_{k_Y}\}$, such that 
$\bigcup_{i=1}^{k_X}{X_i} = X$ and 
$\bigcup_{i=1}^{k_Y}{Y_i} = Y$.
Every $X_i \in \mathcal{P}_X$ (respectively, $Y_i \in \mathcal{P}_Y$) 
has size $|X|/k_X$ (respectively, $|Y|/k_Y$).
$X_1, X_2, \cdots, X_{k_X}$ and $Y_1, Y_2, \cdots, Y_{k_Y}$ 
are called \textit{super nodes}, 
and two super nodes $X_i$ and 
$Y_j$ are \textit{adjacent} if some vertex in $X_i$ 
and some vertex in $Y_j$ are adjacent in $G$. 
If $X_i$ and $Y_j$ are adjacent, then $X_i$ and $Y_j$ form a \textit{super edge}.
In the MIN-REP problem, our task is to choose representatives for super nodes 
so that if $X_i$ and $Y_j$ form a super edge, 
then some representative for $X_i$ and
some representative for $Y_j$ are adjacent in $G$. 
Note that a super node may have multiple representatives.
Specifically, the goal of the MIN-REP problem is to 
find the smallest subset $R \subseteq X \cup Y$ such that if $X_i$ and 
$Y_j$ form a super edge, then $R$ must contain two vertices $x$ 
and $y$ such that $x \in X_i$, $y \in Y_j$ and $(x, y) \in E$. 
In this case, we say that $\{x,y\}$ \textit{covers} 
the super edge $(X_i, Y_j)$. 
The inapproximability result of the MIN-REP
problem is stated as the following theorem.
\begin{theorem}[Kortsarz \textit{et al.}~\cite{doi:10.1137/S0097539702416736}]
\label{MIN-REP}
For any constant $\epsilon > 0$, 
unless $NP \subseteq DTIME(n^{poly\log n})$,
there is no polynomial-time algorithm that can distinguish 
between instances of the MIN-REP problem with a solution of size $k_X+k_Y$ 
and instances where every solution is of size at least 
$(k_X+k_Y) \cdot 2^{\log^{1-\epsilon}{n(G)}}$, where $n(G)$ is the number of 
vertices in the input graph of the MIN-REP problem.
\end{theorem}

\subsection{The Reduction}
Given inputs $G=(X,Y,E)$, $\mathcal{P}_X$, and $\mathcal{P}_Y$ 
of the MIN-REP problem, 
we construct a corresponding graph $G'(G,\mathcal{P}_X, \mathcal{P}_Y)$ 
of the 1-DR-2 problem. 
When $G$, $\mathcal{P}_X$, and $\mathcal{P}_Y$ are clear from the context, 
we simply write $G'$ instead of $G'(G,\mathcal{P}_X, \mathcal{P}_Y)$. 
Initially, $G'=G$. Hence, $G'$ contains $X$, $Y$, and $E$. 
For each super node $X_i$ (respectively, $Y_i$), 
we create two corresponding vertices $px^1_i$ and $px^2_i$ 
(respectively, $py^1_i$ and $py^2_i$) in $G'$. 
If $x$ is in super node $X_i$ (respectively, $y$ is in super node $Y_i$), 
then we add two edges $(x, px^1_i)$ and $(x, px^2_i)$ 
(respectively, $(y, py^1_i)$ and $(y, py^2_i)$) in $G'$. 
If $X_i$ and $Y_j$ form a super edge, then we add two vertices $r^1_{i,j}$ 
and $r^2_{i,j}$ to $G'$, and we add four edges 
$(px^1_i, r^1_{i,j})$, $(r^1_{i,j}, py^1_j)$, 
$(px^2_i, r^2_{i,j})$, $(r^2_{i,j}, py^2_j)$ to $G'$. 
$r^1_{i,j}$ (respectively, $r^2_{i,j}$) is called the \textit{relay} of 
$px^1_i$ and $py^1_j$ (respectively, $px^2_i$ and $py^2_j$).

Before we complete the construction of $G'$, we briefly explain the 
idea behind the construction so far. 
If two super nodes $X_i$ and $Y_j$ form a super edge, 
then $px^I_i$ and $py^I_j$ ($I \in \{1,2\}$) have a common neighbor in $G'$, 
i.e., the relay $r^I_{i,j}$. 
Because $px^I_i$ and $py^I_j$ are not adjacent, 
$px^I_i$ and $py^I_j$ form a target couple.
To transform a solution $D$ of the 1-DR-2 problem to a solution of 
the MIN-REP problem, we need to transform $D$ to another feasible 
solution $D'$ for the 1-DR-2 problem so that none of the relays is chosen, 
and only vertices in $X \cup Y$ are used to connect $px^I_i$ and $py^I_j$. 
This is the reason that we have two corresponding vertices for each super node 
(and thus two relays for each super edge). Under this setting, 
to connect $px^1_i$ to $py^1_j$ and $px^2_i$ to $py^2_j$, 
choosing two vertices in $X \cup Y$ is no worse than choosing the relays.

\begin{figure}[t]
\begin{center} 
\includegraphics[width=11cm]{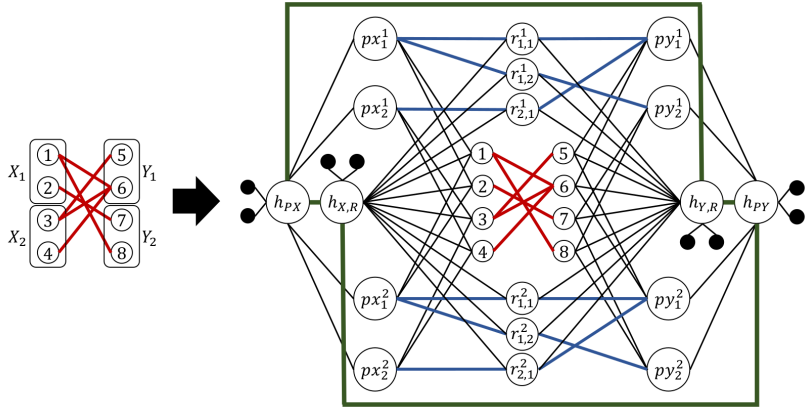} 
\caption{An example of the reduction from the MIN-REP problem to the 
1-DR-2 problem.}
\label{fig: reduction}
\end{center}
\end{figure}

Let $PX = \{px^1_1, px^1_2, \cdots, px^1_{k_X}\} \cup 
\{px^2_1, px^2_2, \cdots, px^2_{k_X}\}$ 
be the set of vertices in $G'$ corresponding to the super nodes in $\mathcal{P}_X$.
Similarly, let $PY = \{py^1_1, py^1_2, \cdots, py^1_{k_Y}\} \cup 
\{py^2_1, py^2_2, \cdots, py^2_{k_Y}\}$. Let $R$ be the set of all relays. 
To complete the construction, we add four vertices (hubs) 
$h_{X,R}$, $h_{Y,R}$, $h_{PX}$, and $h_{PY}$ to $G'$. 
In $G'$, all the vertices in $X$, $Y$, $PX$, and $PY$ are adjacent to 
$h_{X,R}$, $h_{Y,R}$, $h_{PX}$, and $h_{PY}$, respectively. 
Moreover, every relay is adjacent to $h_{X,R}$ and $h_{Y,R}$. 
These four hubs induce a 4-cycle 
$(h_{PX}, h_{Y,R}, h_{PY}, h_{X,R}, h_{PX})$ in $G'$. 
Finally, for each hub $h$, we create two dummy nodes $d_1$ and $d_2$, 
and add two edges $(h, d_1)$ and $(h, d_2)$ to $G'$. 
This completes the construction of $G'$. 
Fig.~\ref{fig: reduction} shows an example of the reduction.
Let $H$ and $M$ be the set of hubs and the set of dummy nodes, respectively.
Hence, the vertex set of $G'$ is 
$X \cup Y \cup PX \cup PY \cup R \cup H \cup M$.
Let $N(u)$ be the set of neighbors of $u$ in $G'$. We then have
\begin{align*}
&N(px) \subseteq X \cup R \cup \{h_{PX}\} \text{ if } px \in PX.
&N(py) \subseteq Y \cup R \cup \{h_{PY}\} \text{ if } py \in PY.\\
&N(x) \subseteq PX \cup Y \cup \{h_{X,R}\} \text{ if } x \in X.
&N(y) \subseteq PY \cup X \cup \{h_{Y,R}\} \text{ if } y \in Y.\\
&N(h_{X,R}) \setminus M = X \cup R \cup \{h_{PX}, h_{PY}\}.
&N(h_{Y,R}) \setminus M = Y \cup R \cup \{h_{PX}, h_{PY}\}.\\
&N(h_{PX}) \setminus M = PX \cup \{h_{X,R}, h_{Y,R}\}.
&N(h_{PY}) \setminus M = PY \cup \{h_{X,R}, h_{Y,R}\}.\\
&N(m) \subseteq H \text{ if } m \in M. 
&N(r) \subseteq PX \cup PY \cup \{h_{X,R}, h_{Y,R}\} 
\text{ if } r \in R.
\end{align*}

Observe that $|R|=O(n(G)^2)$. We have the following lemma.
\begin{lemma}
\label{size}
$n(G')=O(n(G)^2)$.
\end{lemma}

It is easy to check that, for any two adjacent vertices $u$ and $v$ in
$G'$, $u$ and $v$ have no common neighbor. Hence, we have the following lemma.
\begin{lemma}
\label{triangle-free}
$G'$ is triangle-free.
\end{lemma}

We say that a target couple $[a,b]$ is in $[A,B]$ if $a \in A$ and $b \in B$. 
It is easy to verify the following two lemmas.
\begin{lemma}
\label{HMust}
Only $H$ can cover the target couples in $[M, M]$. 
\end{lemma}

\begin{lemma}
\label{DS}
$H$ is a dominating set of $G'$.
\end{lemma}

The proof of the following lemma can be found in the appendix.
\begin{lemma}
\label{H}
$H$ covers all the target couples except those in $[PX,PY]$. 
\end{lemma}

Let $px$ and $py$ be vertices in $PX$ and $PY$, respectively. 
Observe that, if $(px,x,y,py)$ is a path in $G'$, 
then $x \in X$ and $y \in Y$. We then have the following lemma.
\begin{lemma}
\label{CoverPXPY}
$D$ covers target couples $[px^1_i,py^1_j]$ and $[px^2_i,py^2_j]$ 
if and only if at least one of the following conditions is satisfied. 
\begin{enumerate}
\item There exist $x \in X$ and $y \in Y$ such that 
$(px^1_i,x,y,py^1_j)$ and
$(px^2_i,x,y,py^2_j)$
are paths in $G'$ and $\{x,y\} \subseteq D$.
\item $\{r^1_{i,j}, r^2_{i,j}\} \subseteq D$. 
\end{enumerate}
\end{lemma}

\subsection{The Analysis}
Let $I_{MR}$ be an instance of the MIN-REP problem with inputs $G$, 
$\mathcal{P}_X$, and $\mathcal{P}_Y$. 
Let $I_D$ be the instance of the 1-DR-2 problem 
with input $G'(G,\mathcal{P}_X,\mathcal{P}_Y)$. 
To prove the inapproximability result, we use the following two lemmas.
\begin{lemma}
\label{UB}
If $I_{MR}$ has a solution of size $s$, 
then $I_D$ has a solution of size $s+4$.
\end{lemma}
\begin{lemma}
\label{LB}
If every solution of $I_{MR}$ has size at least 
$s \cdot 2^{\log^{1-\epsilon}{n(G)}}$, 
then every solution of $I_D$ has size at least 
$s \cdot 2^{\log^{1-\epsilon}{n(G)}}+4$.
\end{lemma}

\textbf{Proof of Theorem~\ref{thrm: inapprox}:}
By Theorem~\ref{MIN-REP}, for any constant $\epsilon > 0$, 
unless $NP \subseteq DTIME(n^{poly\log n})$,
there is no polynomial-time algorithm that can distinguish 
between instances of the MIN-REP problem with a solution of size $k_X+k_Y$ 
and instances where every solution is of size at least 
$(k_X+k_Y) \cdot 2^{\log^{1-\epsilon}{n(G)}}$.
By the above two lemmas, it is hard to distinguish between instances of 
the 1-DR-2 problem with a solution of size $k_X+k_Y+4$ and instances 
in which every solution is of size at least 
$(k_X+k_Y) \cdot 2^{\log^{1-\epsilon}{n(G)}}+4$.
Therefore, for any constant $\epsilon > 0$, 
unless $NP \subseteq DTIME(n^{poly\log n})$,
there is no polynomial-time algorithm that can approximate the 1-DR-2 problem 
by a factor better than $\frac{(k_X+k_Y)\cdot 2^{\log^{1-\epsilon}{n(G)}}+4}{k_X+k_Y+4}$.
Lemma~\ref{size} implies that, for any constant $\epsilon' > 0$, 
unless $NP \subseteq DTIME(n^{poly\log n})$, there is no 
$O(2^{\log^{1-\epsilon'}{n(G')^{0.5}}})$-approximation algorithm 
for the 1-DR-2 problem. 
By considering sufficiently large instances and a small enough $\epsilon'$, 
we have the hardness result claimed in Theorem~\ref{thrm: inapprox}. 
On the other hand, let 1-DR-$2'$ be the problem obtained by
removing the constraint that any feasible solution 
must be a dominating set from the 1-DR-2 problem.
Thus, in the 1-DR-$2'$ problem, we only focus on covering 
target couples. By Lemmas~\ref{HMust} and~\ref{DS}, 
a solution $D$ is feasible for 
the 1-DR-$2'$ problem with input $G'$ if and only if 
$D$ is a feasible solution of $I_D$.
Thus, the inapproximability result also applies to the 1-DR-$2'$ problem.  
Finally, the proof follows from Lemma~\ref{triangle-free}. \qed

Lemma~\ref{UB} is a direct result of the following claim.
\begin{myclaim}
\label{SHfeasible}
If $S$ is a feasible solution of $I_{MR}$, 
then $S \cup H$ is a feasible solution of $I_D$.
\end{myclaim}
\begin{proof}
Since $H$ is a dominating set, by Lemma~\ref{H}, 
it suffices to prove that every target couple 
$[u,v] = [px^{I_1}_i, py^{I_2}_j]$ in $[PX,PY]$ is covered by $S$.
Note that $[px^{I_1}_i, py^{I_2}_j]$ cannot be a target couple if $I_1 \neq I_2$. 
This is because $px^{I_1}_i$ and $py^{I_2}_j$ do not have a common neighbor 
if $I_1 \neq I_2$.
If $I_1 = I_2$, then the common neighbor must be $r^I_{i,j}$.
By the construction of $G'$, this implies that $X_i$ and $Y_j$ 
form a super edge.
Since $S$ is a feasible solution of $I_{MR}$, there exists $x \in X_i$
and $y \in Y_j$ such that $x$ and $y$ are adjacent in $G$ and $\{x,y\} \subseteq S$. 
Again, by the construction of $G'$, $(u, x, y, v)$ is a path in $G'$. 
Hence, $S \supseteq \{x,y\}$ covers $[u,v]$.
\end{proof}

To prove Lemma~\ref{LB}, we use the following claim.
\begin{myclaim}
\label{XYonly}
$I_D$ has an optimal solution $D^*$, 
such that $D^* \setminus H$ is a feasible solution of $I_{MR}$.
\end{myclaim}

\textbf{Proof of Lemma~\ref{LB}:}
Let $S^*$ be the optimal solution of $I_{MR}$.
By the assumption, we have $|S^*| \geq s \cdot 2^{\log^{1-\epsilon}{n}}$.
It suffices to prove that $S^* \cup H$ is an optimal solution for $I_D$, 
which implies that every feasible solution of $I_D$ has size at least 
$|S^* \cup H| = |S^*| + 4 \geq s \cdot 2^{\log^{1-\epsilon}{n}} + 4$. 
The feasibility of $S^* \cup H$ follows from Claim~\ref{SHfeasible}.
For the sake of contradiction, assume that the optimal solution of $I_D$ has 
size smaller than $|S^* \cup H|=|S^*|+4$. 
Claim~\ref{XYonly} and Lemma~\ref{HMust} then imply that 
$S^*$ is not an optimal solution of $I_{MR}$,
which is a contradiction. \qed

\textbf{Proof of Claim~\ref{XYonly}:}
Let $D_{OPT}$ be any optimal solution of $I_D$.
By Lemmas~\ref{HMust}, \ref{H}, and \ref{CoverPXPY}, 
$D_{OPT} \subseteq H \cup X \cup Y \cup R$.
If $D_{OPT} \cap R = \emptyset$, 
by Lemma~\ref{CoverPXPY}, each target couple $[px^I_i,py^I_j]$ 
is covered by some $x \in X$ and some $y \in Y$. 
By the construction of $G'$,  
such $x$ and $y$ also cover the super edge $(X_i, Y_j)$ in $I_{MR}$.
Because each super edge in $I_{MR}$ has a corresponding target couple in $I_D$, 
$D_{OPT} \setminus H$ is a feasible solution of $I_{MR}$.

If $D_{OPT} \cap R \neq \emptyset$, then some $r^I_{i,j} \in D_{OPT}$.
We can further assume that both $r^1_{i,j}$ and $r^2_{i,j}$ are in $D_{OPT}$;
otherwise, by Lemma~\ref{CoverPXPY}, we can remove $r^I_{i,j}$ from $D_{OPT}$, 
the resulting solution is smaller and is still feasible.
We then replace $r^1_{i,j}$ and $r^2_{i,j}$ with 
some $x \in X$ and some $y \in Y$
satisfying the first condition in Lemma~\ref{CoverPXPY}. 
By Lemma~\ref{CoverPXPY}, the resulting solution is still feasible, 
and the size remains the same.
Repeat the above replacing process until the resulting solution does 
not contain any relay. The proof then follows from the argument 
of the case where $D_{OPT} \cap R = \emptyset$. \qed

\section{Transforming the 1-DR-$\alpha$ Problem to Other Related Problems}
\textbf{Submodular Cost Set Cover Problem:}
The 1-DR-$\alpha$ problem can also be considered as a special case of the
submodular cost set cover problem~\cite{Du2011_submodular, 5438589, Wan2010}.
In the set cover problem, we are given a set of targets $\mathcal{T}$ 
and a set of objects $\mathcal{S}$.
Each object in $\mathcal{S}$ can cover a subset of $\mathcal{T}$ 
(specified in the input).
The goal is to choose the smallest subset of $\mathcal{S}$ that 
covers $\mathcal{T}$. In the submodular cost set cover problem, 
there is a non-negative submodular function $c$ that maps each subset of 
$\mathcal{S}$ to a cost, and the goal is to find the set cover with 
the minimum cost. To transform the 1-DR-$\alpha$ problem with input $G=(V,E)$ 
to the submodular cost set cover problem,
let $\mathcal{T}$ be the union of $V$ and the set of all target couples,
and let $\mathcal{S}$ be the set of all subsets of $V$ with size 
at most $\alpha$. Hence, each object in $\mathcal{S}$ is a subset of $V$.
An object $S \in \mathcal{S}$ can cover a vertex $v$ if 
$v$ is adjacent to some vertex in $S$ or $v \in S$. 
An object $S \in \mathcal{S}$ can cover a target couple $[u,v]$ 
if $m^S(u,v) \leq \alpha$.
The cost of a subset $\mathcal{S}'$ of $\mathcal{S}$ is simply the size 
of the union of objects in $\mathcal{S}'$, i.e., the number of distinct 
vertices specified in $\mathcal{S}'$. 

Iwata and Nagano proposed a $|\mathcal{T}|$-approximation algorithm 
and an $f$-approximation algorithm, where $f$ is the maximum frequency, 
$\argmax_{T\in \mathcal{T}}
{|\{S \in \mathcal{S}| S \text{ covers } T\}|}$~\cite{5438589}.
Koufogiannakis and Young also proposed an $f$-approximation algorithm
when the cost function $c$ is non-decreasing~\cite{Koufogiannakis2013}.
It is easy to see that these algorithms give trivial bounds for the 
1-DR-$\alpha$ problem.
When the cost function $c$ is integer-valued, non-decreasing, and satisfies $c(\emptyset)=0$, 
Wan \textit{et al.} proposed a $\rho H(\gamma)$-approximation algorithm, where 
$\rho = \min\limits_{\mathcal{S}^*: \mathcal{S}^* \text{ is an optimal solution}}
{\frac{\sum_{S \in \mathcal{S}^*}c(\{S\})}{c(\mathcal{S}^*)}}$, 
$\gamma$ is the largest number of targets that can be covered by an object in 
$\mathcal{S}$, and $H(k)$ is the $k$-th Harmonic number~\cite{Wan2010}.
Du \textit{et al.} applied this algorithm to the 1-DR-$\alpha$ problem on UDG
for $\alpha \geq 5$ and obtained a constant factor approximation 
algorithm~\cite{Du2011_submodular}. 
It is unclear whether or not $\rho$ can be upper bounded by $O(n^{1-\epsilon})$ 
for some $\epsilon > 0$ when applied to the 1-DR-$\alpha$ problem on general graphs.

\textbf{Minimum Rainbow Subgraph Problem on Multigraphs:}
Given a set of $p$ colors and a multigraph $H$, 
where each edge is colored with one of the $p$ colors,
the Minimum Rainbow Subgraph (MRS) problem 
asks for the smallest vertex subset $D$ of $H$, 
such that each of the $p$ colors appears in some edge induced by $D$.
The 1-DR-2 problem can be transformed to the MRS problem as follows.
Let $G=(V,E)$ be the input graph of the 1-DR-2 problem.
Let $T$ be the union of $V$ and the set of all target couples.
The set of colors for the MRS problem is $\{c_i|i \in T\}$.
The input multigraph $H$ of the MRS problem has the same vertex set as $G$. 
To form a dominating set, for each $v \in V$, 
$v$ is incident to $d(v)+1$ loops $(v,v)$ in $H$, 
where $d(v)$ is the degree of $v$ in $G$.
Each of these loops receives a different color in 
$\{c_v\} \cup \{c_u| (u,v) \in E\}$.
For each target couple $[u,v]$ in $G$, if $w$ is a common 
neighbor of $u$ and $v$ in $G$, 
we add a loop $(w,w)$ with color $c_{[u,v]}$ to $H$.
Finally, for each target couple $[u,v]$ in $G$, 
if $(u,w_1, w_2,v)$ is a path in $G$, 
we add an edge $(w_1, w_2)$ with color $c_{[u,v]}$ to $H$.
The MRS problem can be transformed to the SCP problem.
When the input graph is simple, Tirodkar and Vishwanathan proposed an 
$O(n^{1/3}\log n)$-approximation algorithm~\cite{Tirodkar2017}.

\bibliographystyle{abbrv}

\appendix
\section{Proof of Claim~\ref{c}}
\begin{align*}
n^{1-c} &= \sqrt{n \cdot \alpha(n^c)^{\alpha-2}\ln t} \\
\Leftrightarrow n^{2-2c} &= n \cdot \alpha(n^c)^{\alpha-2}\ln t 
\text{ (both sides are non-negative)} \\
\Leftrightarrow n^{2-2c-(1+c(\alpha-2))} &= \alpha \ln t \\
\Leftrightarrow n^{1-c\alpha} &= \alpha \ln t.
\end{align*}
When $c = \frac{1}{\alpha}-\frac{\ln \ln (t^{\alpha})}{\alpha \ln n}$, 
\begin{align}
n^{1-c\alpha} &= n^{1-(1-\frac{\ln \ln (t^{\alpha})}{\ln n})} \nonumber \\ 
&= n^{\frac{\ln \ln (t^{\alpha})}{\ln n}} \label{eq0} \\ 
&= (n^{\ln (\ln (t^{\alpha}))})^{\frac{1}{\ln n}} \label{eq1} \\ 
&= ((\ln (t^{\alpha}))^{\ln n})^{\frac{1}{\ln n}} \label{eq2} \\
&= (({\alpha}\ln t)^{\ln n})^{\frac{1}{\ln n}} \label{eq3} \\
&= {\alpha}\ln t \label{eq4}.
\end{align}
Hence, when $c = \frac{1}{\alpha}-\frac{\ln \ln (t^{\alpha})}{\alpha \ln n}$, 
$n^{1-c}=\sqrt{n \cdot \alpha(n^c)^{\alpha-2}\ln t}$.

Finally, when $c = \frac{1}{\alpha}-\frac{\ln \ln (t^{\alpha})}{\alpha \ln n}$,
\begin{align*}
n^{1-c} &= n^{1-\frac{1}{\alpha}+\frac{\ln \ln (t^{\alpha})}{\alpha \ln n}}\\
&= n^{1-\frac{1}{\alpha}} \cdot n^{\frac{\ln \ln (t^{\alpha})}{\alpha \ln n}}\\
&= n^{1-\frac{1}{\alpha}} 
\cdot (n^{\frac{\ln \ln (t^{\alpha})}{\ln n}})^{\frac{1}{\alpha}}\\
&= n^{1-\frac{1}{\alpha}} \cdot (\alpha \ln t)^{\frac{1}{\alpha}}.
\end{align*}
In the last equality, we reuse Eq.\eqref{eq0}-Eq.\eqref{eq4}. \qed

\section{Proof of Lemma~\ref{H}}
If $[u,v]$ is in $[PX, PX \cup \{h_{X,R}, h_{Y,R}\}]$, 
                 $[PY, PY \cup \{h_{X,R}, h_{Y,R}\}]$, 
                 $[X, X \cup R \cup \{h_{PX}, h_{PY}\}]$,
                 $[Y, Y \cup R \cup \{h_{PX}, h_{PY}\}]$,
                 or $[R, R \cup \{h_{PX}, h_{PY}\}]$, 
                 then $[u,v]$ can be covered by one vertex in $H$.
If $[u,v]$ is in $[PX,Y], [PY,X], [X, \{h_{Y,R}\}]$, or $[Y,\{h_{X,R}\}]$, 
                 then $[u,v]$ can be covered by an edge in $H$.
If $[u,v]$ is in $[PX, X \cup R \cup \{h_{PX}, h_{PY}\}]$, 
                 $[PY, Y \cup R \cup \{h_{PX}, h_{PY}\}]$,
                 or $[X,Y]$,  
                 then $[u,v]$ cannot be a target couple 
                 (since $u$ and $v$ do not have a common neighbor). 
If $[u,v]$ is in $[X,\{h_{X,R}\}], [Y,\{h_{Y,R}\}]$, 
                 or $[R, \{h_{X,R}, h_{Y,R}\}]$, 
                 then $[u,v]$ cannot be a target couple 
                 (since $u$ and $v$ are adjacent\footnote{In addition, 
                 by Lemma~\ref{triangle-free}, 
                 $u$ and $v$ do not have a common neighbor.}).
Moreover, it is easy to see that $H$ covers all the target couples in $[H,H]$ or
$[V(G'), M]$, where $V(G')$ is the vertex set of $G'$.
Finally, observe that if $[u,v]$ is in $[PX,PY]$, 
then $H$ cannot cover $[u,v]$. \qed
\end{document}